%% file: main.tex
\documentclass{article}

\usepackage[english]{babel}

\usepackage[letterpaper,top=2cm,bottom=2cm,left=3cm,right=3cm,marginparwidth=1.75cm]{geometry}

\usepackage{amsmath}
\usepackage{graphicx}
\usepackage[colorlinks=true, allcolors=blue]{hyperref}
\input{_preamble}

\title{Is Nash Equilibrium Approximator Learnable?}
\author{
\textbf{Zhijian Duan}$^{1}$, 
\textbf{Wenhan Huang}$^{2}$,
\textbf{Dinghuai Zhang}$^{3}$, \\
\textbf{Yali Du}$^{4}$,
\textbf{Jun	Wang}$^{5}$,
\textbf{Yaodong	Yang}$^{6}$,
\textbf{Xiaotie Deng}$^{1,6}$ 
\\
$^{1}$CFCS, School of Computer Science, Peking University
$^{2}$Shanghai Jiao Tong University \\
$^{3}$Mila - Quebec AI Institute
$^{4}$King's College London
$^{5}$University College London \\
$^{6}$CMAR, Institute for AI, Peking University
\\
\texttt{zjduan@pku.edu.cn},~
\texttt{rowdark@sjtu.edu.cn},\\
\texttt{dinghuai.zhang@mila.quebec},~
\texttt{yali.du@kcl.ac.uk},~
\texttt{jun.wang@cs.ucl.ac.uk},\\
\texttt{\{yaodong.yang, xiaotie\}@pku.edu.cn}
}
\date{}

\begin{document}
\maketitle

\input{_abstract}

\input{_main_body}

\bibliographystyle{plainnat}
\bibliography{_reference}

\input{_appendix.tex}
\end{document}

%% file: _preamble.tex
\usepackage{booktabs}
\usepackage{amsfonts}
\usepackage{amsmath}
\usepackage{multirow}
\usepackage{amsthm}
\usepackage{diagbox}
\usepackage{subcaption}
\usepackage{bm}
\usepackage{color}
\usepackage{graphicx}
\usepackage{natbib}
\usepackage{enumitem}
\usepackage{algorithm}
\usepackage{algorithmic}
\usepackage{mathtools}
\usepackage{thmtools}
\usepackage{balance}
\usepackage{comment}
\usepackage[capitalize,noabbrev,nameinlink]{cleveref}

\usepackage[textsize=tiny]{todonotes}

\theoremstyle{plain}
\newtheorem{theorem}{Theorem}[section]

\newtheorem{lemma}[theorem]{Lemma}
\newtheorem{corollary}[theorem]{Corollary}

\theoremstyle{definition}
\newtheorem{definition}[theorem]{Definition}
\newtheorem{assumption}[theorem]{Assumption}
\theoremstyle{remark}
\newtheorem{remark}[theorem]{Remark}

\newcommand{\EE}{\mathbb{E}}
\newcommand{\PP}{\mathbb{P}}
\newcommand{\RR}{\mathbb{R}}

\newcommand{\Hh}{\mathcal{H}}
\newcommand{\Ff}{\mathcal{F}}

\newcommand{\Nn}{\mathcal{N}}
\newcommand{\Aa}{\mathcal{A}}
\newcommand{\Dd}{\mathcal{D}}

\newcommand{\Uu}{\mathcal{U}}

\newcommand{\Rr}{\mathcal{R}}
\newcommand{\Nap}{\textsc{NashApr}}

\newcommand{\Dim}{{n|A|}}

\newcommand{\Lip}{{Lipschitz~}}
\newcommand{\norm}[1]{{\|#1\|}}
\newcommand{\Rom}[1]{{\uppercase\expandafter{\romannumeral#1}}}

%% file: _abstract.tex
\begin{abstract}
In this paper, we investigate the learnability of the function approximator that approximates Nash equilibrium (NE) for games generated from a distribution.
First, we offer a generalization bound using the Probably Approximately Correct (PAC) learning model.
The bound describes the gap between the expected loss and empirical loss of the NE approximator.
Afterward, we prove the agnostic PAC learnability of the Nash approximator.
In addition to theoretical analysis, we demonstrate an application of NE approximator in experiments.
The trained NE approximator can be used to warm-start and accelerate classical NE solvers.
Together, our results show the practicability of approximating NE through function approximation.
\end{abstract}

%% file: _main_body.tex
\section{Introduction}
Nash equilibrium (NE)~\citep{nash1950equilibrium}, in which each agent's strategy is optimal given the strategies of all other agents, is one of the most important solution concepts in game theory.
It can be used to analyze the outcome of strategic interactions among rational agents.
An NE or $\epsilon$-approximate Nash equilibrium ($\epsilon$-NE) strategy can also be a good guide for agents in the game since agents have no or negligible incentive to disobey individually.
There has been increasing interest in NE due to its broad applications in Generative Adversarial Networks (GAN)~\citep{goodfellow2014generative}, Multi-Agent Reinforcement Learning (MARL)~\citep{yang2020overview}, multi-agent systems~\citep{shoham2008multiagent}, economics~\citep{deng2017learn,deng2019bayesian}, and online advertising~\citep{deng2020private}.
Although NE always exists in normal-form games~\citep{nash1950equilibrium}, finding an NE is PPAD-complete even for $2$-player games~\citep{chen2009settling} and $3$-player games~\citep{daskalakis2009complexity}.
Such negative results lead to increased attention on developing algorithms to approximate NE.

While many algorithms were proposed to find $\epsilon$-NE for some approximation  $\epsilon>0$~\citep{kontogiannis2006polynomial,daskalakis2009note, daskalakis2007progress, czumaj2019distributed, bosse2007new, kontogiannis2007efficient, TS0.3393, DFM1/3}, these works focus on solving a single game in isolation.
However, many similar games usually need to be solved in practice or in some multi-agent learning algorithms. 
For instance, in repeated contextual games such as traffic routing~\citep{sessa2020contextual}, the utility function depends on contextual information generated from a distribution.
The Nash Q-learning~\citep{hu2003nash} algorithm, which solves Markov games via value-based reinforcement learning, needs to compute NE for a normal-form game every time it updates the Q-value.
In these settings, traditional solvers have to compute from scratch for every game, ignoring the similarity among those games.
As an improvement, it can be preferable to construct a function approximator that predicts NE from game utility~\citep{marris2022turbocharging, feng2021neural}.
The NE approximator is trained through the historical data and can provide an approximate solution quickly at the test time.

Several critical theoretical issues arise in developing algorithms to predict NE from samples.
First, the NE approximator is learned from training data and will be evaluated by unseen games in testing.
Therefore, its generalization ability, i.e., its performance in testing, needs to be clarified.
Moreover, people also care about the sample complexity (how many training samples we need) to get a reasonable approximator.

In this paper, we make the first step to study the learnability of predicting NE by function approximation.
We consider general $n$-player normal-form games with fixed action space.
We follow the standard Probably Approximately Correct (PAC) learning~\citep{valiant1984theory,haussler1990probably} model, in which game utilities are independently generated from an identical distribution, both in training and testing.
One challenge is the non-uniqueness issue of exact NE, which brings difficulty for naively adopting supervised learning techniques.
Inspired by the definition of $\epsilon$-NE, we set up a self-supervised loss function to evaluate the performance of an NE approximator.
Such Nash approximation loss is Lipschitz continuous to game utility and players' strategies.
Based on that, we present a generalization bound for any NE approximators.
The bound provides a confidence interval on the expected loss based on the empirical loss in training. 
Furthermore, based on a mild assumption of the NE approximator function class, we prove that it is agnostic PAC learnable to predict NE from samples.
To the best of our knowledge, this is the first result that addresses the PAC learnability of Nash equilibrium.

In addition to the theoretical analysis, we demonstrate a practical application of the learned NE approximator.
We show that it can warm-start other classic approximate NE solvers.
By doing so, we combine both advantages of the function approximation method and the traditional approach.
The former helps to provide an effective initial solution in batches with low computational costs, and the latter provides theoretical guarantees.
Specifically, we conduct numerical experiments in bimatrix games.
We train a neural network-based NE approximator and use the predicted solutions as the pre-solving initialization for the algorithm of \citet{TS0.3393} and the start-of-the-art approximate NE solver proposed by \citet{DFM1/3}.
In both cases, we report faster convergence.

Our paper is organized as follows: 
In \cref{sec:related work} we describe related works;
In \cref{sec:preliminary} we introduce the preliminary of game theory; 
In \cref{sec:algorithm} we set up the PAC learning framework for predicting NE from samples; 
We present our learnability results in \cref{sec:theory};
We conduct numerical experiments and demonstrate the application in \cref{sec:experiments};
We draw our conclusion in \cref{sec:conclusion}.

\section{Related Work}
\label{sec:related work}

\paragraph{Classic solvers with feasibility guarantee} For $2$-player games, there are algorithms with a theoretical guarantee for maximum Nash approximation loss (See definition in \cref{eq:Nap}).
\citet{kontogiannis2006polynomial} and \citet{daskalakis2009note} introduced simple polynomial-time algorithms based on searching small supports to reach an approximation loss of $3/4$ and $1/2$, respectively. 
\citet{daskalakis2007progress} provided an algorithm of approximation loss $0.38$ by enumerating arbitrarily large supports, and this approximation loss is also achieved by  \citet{czumaj2019distributed} with a different approach.  
\citet{bosse2007new} proposed an algorithm based on \citet{kontogiannis2007efficient} to reach an approximation loss of $0.36$. 
TS algorithm~\citep{TS0.3393} achieves an approximation loss of $0.3393$, and \citet{chen2021tightness} proved that the bound is tight.
Recently, DFM algorithm~\citep{DFM1/3}, an improved version of \citet{TS0.3393}, establishes the best currently known approximation loss of $1/3$. 
However, computing approximate NE for even arbitrary constant approximation is PPAD-hard~\citep{daskalakis2013complexity}.

\paragraph{Learning approaches} Learning is another paradigm to compute approximate NE by repeatedly proposing temporal strategies and updating them with feedback rewards.
Fictitious play~\citep{monderer1996fictitious} is the most well-known learning-based algorithm to approximate NE, and \citet{conitzer2009approximation} proves that it reaches an approximate loss of $1/2$ when given constant rounds. 
Double Oracle methods~\citep{mcmahan2003planning,dinh2022online} and PSRO methods~\citep{lanctot2017unified, perez2021modelling}, though effective, target solving zero-sum games only. 
Online learning methods, including regret matching~\citep{hart2000simple}, Hedge~\citep{auer1995gambling} and Multiplicative weight update~\citep{arora2012multiplicative}, are proved to converge to (approximate) coarse correlated equilibrium~\citep{cesa2006prediction}.

\paragraph{Data-driven approaches} In addition to traditional methods, many works have proposed to approximate NE through data-driven approaches.
Some of them make use of the historical game-playing data and learn the game utility functions~\citep{bertsimas2015data, zhang2017data, allen2022using} or game gradients~\citep{ling2018game, ling2019large, heaton2021learn} from the observed (approximate) NE.
By doing so, they can predict approximate NE solutions for a class of games (e.g., contextual games~\citep{heaton2021learn,sessa2020contextual}).
Another way is to learn a function approximator that maps game utility to an approximate solution~\citep{marris2022turbocharging}.
Such NE approximator can be applied in PSRO~\citep{feng2021neural}.
Recently, ~\citet{harris2023metalearning} introduce meta-learning algorithms for equilibrium finding.
In our paper, we study the generalization ability of the NE approximator and the PAC learnability of NE.

\paragraph{Learnability} 
As for learnability analysis in games, \citet{viqueira2019learning} and \citet{marchesi2020learning} provide the PAC analysis of learning the game utility in simulation-based games, in which the utility is obtained by query and would potentially be disturbed by noise.
A \emph{Nash Oracle}, which can output the exact NE for arbitrary games directly, is assumed in these papers.
Similarly, \citet{fele2020probably} considers games with noisy utilities and studies the learnability of NE, given the strong assumption of Nash Oracle.
As a comparison, we do not assume any Nash Oracles in our paper.
Some other works consider query complexity to approximate NE~\citep{fearnley2015learning, fearnley2016finding}, while we focus on the sample complexity to learn a generalizable NE approximator.
Moreover, while \citet{jin2021v} and \citet{bai2020near} propose PAC learnable algorithm to approximate NE in a zero-sum Markov game, and to approximate Coarse Correlated Equilibria (CCE) or Correlated Equilibria (CE) in a general-sum Markov game, we must highlight the difference that we consider the PAC analysis of NE in \emph{general-sum} games sampled from a same \emph{arbitrary distribution}, instead of approximating NE for one specific game instance.

\section{Game Theory Preliminaries}
\label{sec:preliminary}

\paragraph{Normal-form games}
We denote a normal-form game with joint utility function $u$ as $\Gamma_u = (N, A, u)$ and explain each item as follows.
\begin{itemize}
    \item 
$N = \{1, 2, \dots, n\}$ is the set of all the $n$ players. 
Each player is represented by the index $i \in N$. 
    \item 
$A = A_1 \times A_2 \times \dots \times A_n$ is the combinatorial action space of all players, in which $A_i$ is the action space for player $i$. 
For player $i\in N$, let $a_i \in A_i$ be a specific action and $|A_i|$ be the number of actions (An action is also referred to as a pure strategy).
An action profile $a = (a_1, a_2, \dots, a_n) \in A$ represents one play of the game in which the player $i$ takes her corresponding action $a_i \in A_i$.
The action space $A$ is a Cartesian product that contains all possible action profiles. 
Therefore, we have $|A| = \prod_{i\in N}|A_i|$.
\item 
$u = (u_1, \dots, u_n)$ is the game utility (payoff), in which $u_i: A \to \RR$ is the utility function (or utility matrix, equivalently) for player $i$.
$u_i$ describes the utility of player $i$ on each possible action profile $a = (a_1, a_2, \dots, a_n) \in A$.
We have $|u_i| = |A|$ and $|u| = n|A|$.
In our paper, we assume each utility is in the range of $[0, 1]$ without loss of generality. 
Such an assumption is widely-used in previous literatures~\citep{TS0.3393,DFM1/3}.
\end{itemize}

A mixed strategy of player $i$, denoted by $\sigma_i$, is a distribution over her action set $A_i$.
Specifically,  $\sigma_i(a_i)$ represents the probability that player $i$ chooses action $a_i$.
Under such definition, we have $\sum_{a_i\in A_i}{\sigma_i(a_i)}=1$. 
Denote $\Delta A_i \ni \sigma_i$ be the set of all the possible mixed strategies for player $i$.
A mixed strategy profile $\sigma=(\sigma_1,\sigma_2,\dots,\sigma_n)$ is a joint strategy for all the players.
Based on $\sigma$, the probability of action profile $a = (a_1, a_2, \dots, a_n)$ being played is $\sigma(a) := \prod_{i\in N} \sigma_i(a_i)$.
Notice that an action profile $a$ (i.e., a pure strategy profile) can also be seen as a mixed strategy profile $\sigma$ with $\sigma_i(a_i) = 1$ for all $i\in N$.
The expected utility of player $i$ under $\sigma$ is 
\begin{equation*}
    u_i(\sigma) = \mathbb{E}_{a \sim \sigma}[u_i(a)] = \sum_{a\in A} \sigma(a) u_i(a).
\end{equation*}
Besides, on behalf of player $i$, the other players' strategy profile is denoted as $\sigma_{-i}=(\sigma_1,\dots,\sigma_{i-1},\sigma_{i+1},\dots,\sigma_n)$.

\paragraph{($\epsilon$-approximate) Nash equilibrium}
Nash equilibrium is one of the most important solution concepts in game theory.
A (mixed) strategy profile $\sigma^* = (\sigma^*_1, \sigma^*_2, \dots, \sigma^*_n)$ is called a {Nash equilibrium} if and only if for each player $i \in N$,
her strategy is the best response given the strategies $\sigma^*_{-i}$ of all the other players. 
Formally, 
\begin{align}
	\tag{NE}
	u_i(\sigma_i,\sigma^*_{-i}) &\le u_i(\sigma^*_i, \sigma^*_{-i}),\quad\forall i \in N, \sigma_i \in \Delta A_i
\end{align} 
However, computing NE for even general $2$-player or $3$-player games is PPAD-hard~\citep{chen2009settling, daskalakis2009complexity}.
Given such hardness, many works focus on finding approximate solutions.
For arbitrary $\epsilon > 0$, we say a strategy profile $\hat{\sigma}$ is an \emph{$\epsilon$-approximate Nash equilibrium} ($\epsilon$-NE) if no one can achieve more than $\epsilon$ utility gain by deviating from her current strategy.
Formally, 
\begin{align}
	\tag{$\epsilon$-NE}
	u_i(\sigma_i, \hat{\sigma}_{-i}) &\le u_i(\hat{\sigma}_i, \hat{\sigma}_{-i}) + \epsilon,\quad \forall i \in N, \sigma_i \in \Delta A_i
\end{align}
The definition of $\epsilon$-NE reflects the idea that players might not be willing to deviate from their strategies when the amount of utility they could gain by doing so is tiny (not more than $\epsilon$).

\section{Learning Framework}

\label{sec:algorithm}
In this section, we set up the PAC learning framework of predicting NE in $n$-player normal-form games with fixed players and fixed action space.
The learning framework includes a domain set $\Uu$, a game-generation distribution $\Dd$, a hypothesis class $\Hh$ of the NE approximator, a training set $S$, and evaluation metrics to evaluate the performance of any NE approximators. 

Domain set is defined as the set of all the possible input games.
In our paper, the domain set $\Uu$ includes all the possible game utilities given the fixed players and action space.
Following the standard PAC learning paradigm, we assume each game utility $u \in \Uu$ is sampled independent and identically from a game-generation distribution $\Dd$ with $\operatorname{supp}(\Dd) \subseteq \Uu$.
The generated games may belong to a specific game class (e.g., symmetric games).
We make no assumption about $\Dd$.
The learner does not know the exact form of $\Dd$, but she can access the generated samples. 

The learner should choose in advance (before seeing the data) a class of functions $\Hh$, where each function $h \colon \Uu \to \Delta A_1 \times \Delta A_2 \times \dots \times \Delta A_n$ in $\Hh$ maps a game utility to a joint strategy of $n$ players.
We call such function class $\Hh$ the \emph{hypothesis class}.
In our paper, we consider hypothesis classes with infinite size.
We will describe how we measure the capacity of $\Hh$ in \cref{sec:theory}.
During learning, a \emph{training set} $S$ of size $m$ is provided to the learner. 
$S= \{u^{(1)}, u^{(2)}, \dots, u^{(m)}\}$ contains $m$ game utilities drawn i.i.d. from domain set $\Uu$ according to $\Dd$.


\begin{table}[t]
	    \centering
	    \caption{An example illustrating the non-uniqueness issue of exact NE, in which $A_1 = \{L, R\}$ and $A_2 = \{U, D\}$.
		    Each element $(x, y)$ in the table represents $u_1(\cdot, \cdot) = x$ and $u_2(\cdot, \cdot) = y$ for the corresponding joint action profile.
		    There are two pure NE (bolded) and one mixed NE in the example.
		    }
	    \vskip 5pt
	    \begin{tabular}{c
			    cc}
		        \toprule
		        & $U$ & $D$ \\
		        \midrule
		        $L$ & $(0, 0)$ & $\bm{(1, 0.5)}$ \\
		        \midrule
		        $R$ & $\bm{(0.5, 1)}$ & $(0, 0)$ \\
		        \bottomrule
		    \end{tabular}
	    \label{tab:non-unique}
	\vspace{-0.2cm}
	\end{table}

One challenge for learning to predict NE is the \emph{non-uniqueness issue}: There may be multiple NEs for a game (See \cref{tab:non-unique} for an illustration).
Such an issue brings trouble for applying supervised learning.
The equilibrium selection problem is nontrivial and many works made some assumptions to ensure the uniqueness of NE~\citep{bertsimas2015data,zhang2017data,li2020end,heaton2021learn}.
To deal with the issue, we use Nash approximation loss to measure the level of approximation of a mixed strategy to NE.
The metrics is widely used in the literature for approximating NE\citep{TS0.3393,DFM1/3}.
Nash approximation is defined as follows
\footnote{
	Another similar concept is called Nash exploitability~\citep{lockhart2019computing}: $\textsc{NashExpl}_i(\sigma, u) := \max_{\sigma'_i \in \Delta A_i}u_i(\sigma'_i, \sigma_{-i}) - u_i(\sigma)$.
}:
\begin{definition}[Nash approximation, $\Nap$]
	\label{def:Nap}
	For normal-form game $\Gamma_u = (N, A, u)$, the Nash approximation loss of strategy profile $\sigma$ with respect to game utility $u$ is the maximum utility gain each player can obtain by deviating from her strategy. Formally, 
	\begin{equation}
		\label{eq:Nap}
		\begin{aligned}
			\Nap(\sigma, u):=& \max_{i\in N}\max_{\sigma'_i \in \Delta A_i}  [u_i(\sigma'_i, \sigma_{-i}) - u_i(\sigma_i, \sigma_{-i})]
			\\
			=& \max_{i\in N}\max_{a_i \in A_i}  [u_i(a_i, \sigma_{-i}) - u_i(\sigma_i, \sigma_{-i})].
		\end{aligned}
	\end{equation}
\end{definition}

The computation of $\Nap(\sigma, u)$ only involves $\sigma$ and $u$.
Thus we do not need any NE or side information.
Besides, as we will discuss in \cref{sec:theory}, the Nash approximation loss is \Lip continuous with respect to both inputs, which helps to derive our results. 


For finite $\Hh$, it is trivial to provide a PAC learnable result~\citep{shalev2014understanding}. 
The learning algorithm $\Aa: \Uu^m \to \Hh$ aims to learn a good NE approximator $h \in \Hh$ from the training data $S$, aiming to minimize the \emph{true risk} $L_\Dd(h)$ of using $h$. 
The true risk is the expected Nash approximation of $h$ under distribution $\Dd$:
\begin{equation}
	\label{eq:L_D}
	L_\Dd(h) := \EE_{u\sim \Dd}\Big[\Nap(h(u), u)\Big],
\end{equation}
We also define the \emph{empirical risk} $L_S(h)$ on the data set $S$ as: 
\begin{equation}
	\label{eq:L_S}
	L_S(h) := \frac{1}{|S|}\sum_{u\in S} \Nap(h(u), u)
\end{equation}

Given enough samples, the true risk can be estimated by the empirical risk (See \cref{theorem:GB}).
Therefore, \emph{empirical risk minimization} (ERM) can be applied to learn an NE approximator $h$ from hypothesis class $\Hh$: 
\begin{equation}
	\label{eq:ERM}
	\mathrm{ERM}_\Hh(S) \in \arg\min_{h\in \Hh}L_S(h)
\end{equation}


\begin{algorithm}[t]
	\caption{NE approximator Learning via minibatch SGD}
	\label{alg:training}
	\begin{algorithmic}[1]
		\STATE {\bfseries Input:} Training set $S$ of size $m$
		\STATE {\bfseries Parameters:} Number of iterations $T > 0$, batch size $B > 0$, learning rate $\eta > 0$, initial parameters $w_0 \in \RR^d$ of the NE approximator model.
		\FOR{$t~=~0$ \textbf{to} $T$}
		\STATE Receive minibatch ${S}_t \,=\, \{u^{(1)}, \ldots, u^{(B)}\} \subset S$ 
		
		\STATE Compute the empirical average loss of $S_t$:
		\STATE ~~~~$L_{S_t}(h^{w_t}) \gets \frac{1}{B}\sum_{i=1}^B \Nap(h^{w_t}(u^{(i)}), u^{(i)}) $
		
		\STATE Update model parameters:
		\STATE ~~~~$w_{t+1} \gets w_t - \eta\nabla_{w_t} L_{S_t}(h^{w_t}) $
		\ENDFOR
	\end{algorithmic}
\end{algorithm}

However, in practice, it is usually intractable to implement the ERM algorithm, especially when $\Hh$ is infinite.
Following the standard approach in deep learning community~\citep{goodfellow2016deep}, we can approximate ERM by \emph{minibatch Stochastic Gradient Descent} (minibatch SGD).
Specifically, we parameterize the NE approximator as $h^w$ with $d$-dimensional parameter variable $w \in \RR^d$  (e.g., the weights of a neural network).
We optimize $w$ by the standard minibatch SGD algorithm (See \cref{alg:training}).
This method is feasible since $\Nap(\sigma, u)$ is differentiable almost everywhere, except for some minor points on a zero-measure set.
Those minor points appear when one of the two maximum operations in $\Nap(\sigma, u)$ has multiple maximum inputs.
We can set one of them according to any tie-breaking rule as the outcome to compute the corresponding gradient. 

We must emphasize that the empirical risk minimization algorithm will only be used in the PAC learnability analysis in \cref{sec:agnostic_PAC}.
Moreover,  \cref{alg:training} will only be used to demonstrate the application of NE approximator in \cref{sec:experiments}. 
The generalization bound in \cref{sec:GB} is unrelated to the learning algorithm we apply.
Instead, we can use any learning algorithms such as the approach in \citet{marris2022turbocharging} to obtain the NE approximator, and the generalization bound still holds. 

\section{Theoretical Learnability Results}
\label{sec:theory}
In this section, we present our theoretical learnability result for predicting NE from samples.
We first analyze the Lipschitz property of Nash approximation loss. 
Based on that, we provide a generalization bound for the NE approximator.
We further show that Nash equilibrium is agnostic PAC learnable under a mild assumption on $\Hh$.
All the omitted proofs are presented in Appendix.


\subsection{\Lip Property of Nash Approximation}

We start with deriving the Lipschitz continuity of $\Nap(\sigma, u)$ with respect to its first input: the joint strategy profile $\sigma$.
We get the following lemma, which indicates that $\Nap(\sigma, u)$ is $2$-\Lip continuous with respect to $\sigma$ under $\ell_1$-distance.
\begin{restatable}{lemma}{LemLSigma}
	\label{lemma:L_sigma}
	For arbitrary strategy profile $\sigma$ and $\sigma'$, we have 
	\begin{equation*}
		\begin{aligned}
			|\Nap(\sigma,u) - \Nap(\sigma', u)|
			\le 2\norm{\sigma - \sigma'}_1,
		\end{aligned}
	\end{equation*}
	where
	\begin{equation*}
		\norm{\sigma - \sigma'}_1 := \sum_{i \in N}\sum_{a_i \in A_i} |\sigma_i(a_i) - \sigma'_i(a_i)|
	\end{equation*}    
	is the $\ell_1$-distance between two mixed strategy profiles $\sigma, \sigma' \in \Delta A_1 \times \Delta A_2 \times \dots \times \Delta A_n$.
\end{restatable}

We also analyze the Lipschitz property of $\Nap(\sigma, u)$ with respect to the game utility $u$, and get the following result:
\begin{restatable}{lemma}{LemLU}
    \label{lemma:L_u}
	For strategy profile $\sigma$ and arbitrary normal-form game $\Gamma_u = (N, A, u)$ and $\Gamma_v = (N, A, v)$ with $u, v \in \Uu$, we have
	\begin{equation*}
		\begin{aligned}
			|\Nap(\sigma, u) - \Nap(\sigma, v)| \le& 2 \|u - v\|_{\max},
		\end{aligned}
	\end{equation*}
	where
	\begin{equation*}
		\norm{u-v}_{\max} := \max_{i\in N}\max_{a\in A}|u_i(a) - v_i(a)|
	\end{equation*}
	is the $\ell_{\max}$-distance between game utilities $u$ and $v$. 
\end{restatable}

\begin{table}[t]
	\centering
	\caption{An example illustrating the non-smooth issue of exact NE, in which $A_1 = \{L, R\}$ and $A_2 = \{U, D\}$.
		Each element $(x, y)$ in the table represents $u_1(\cdot, \cdot) = x$ and $u_2(\cdot, \cdot) = y$ for the corresponding joint action.
		Minor changes in the utility of game $\Gamma_u$ (into game $\Gamma_v$) can cause different exact NE solutions.
		(a): Game $\Gamma_u$. The unique NE is $(L, U)$.
		(b): Game $\Gamma_v$. The unique NE is $(R, U)$.
		The only difference between $\Gamma_u$ and $\Gamma_v$ is the utility $u_1(R, U)$, which only differs by arbitrary small $2\epsilon$.
	}
	\begin{sc}
		\begin{subtable}[h]{0.45\linewidth}
			\centering
			\caption{}
			\begin{tabular}{ccc}
				\toprule
				& U & D \\
				\midrule
				L & $\bm{(0.5, 0.5)}$ & $(1, 0)$ \\
				\midrule
				R & $(0.5-\epsilon, 1)$ & $(0, 0)$ \\
				\bottomrule
			\end{tabular}
		\end{subtable}
	\begin{subtable}[h]{0.45\linewidth}
			\centering
			\caption{}
			\begin{tabular}{ccc}
				\toprule
				& U & D \\
				\midrule
				L & $(0.5, 0.5)$ & $(1, 0)$ \\
				\midrule
				R & $\bm{(0.5+\epsilon, 1)}$ & $(0, 0)$ \\
				\bottomrule
			\end{tabular}
		\end{subtable}
	\end{sc}
\label{tab:non-smooth}
\end{table}

\begin{remark}
While minor changes in game utility may cause different NEs (See \cref{tab:non-smooth} for illustration of such non-smooth issue), as we can see in \cref{lemma:L_u} the Nash approximation loss is $2$-Lipschitz continuous with respect to the game utility.
Such a continuity result plays a critical role in the theoretical analysis of game utility learning in simulation-based games~\citep{viqueira2019learning}, in which the goal is to recover the actual game utility through the noisy query data and to compute an approximate NE for the underlying game.
\end{remark}

\subsection{Generalization Bound}
\label{sec:GB}
We measure the generalizability of an NE approximator by generalization bound.
Such a bound depends on the complexity of hypothesis class $\Hh$.
We characterize such complexity through (external) covering numbers~\citep{shalev2014understanding}, a standard technique in PAC analysis~\citep{anthony1999neural}.
We first define the distance between two different approximators.

\begin{definition}[$\ell_{\infty, 1}$-distance]
The $\ell_{\infty, 1}$-distance between two NE approximators $h_1, h_2$ is:
\begin{equation*}
	\norm{h_1 - h_2}_{\infty, 1} := \max_{u \in \Uu}\norm{h_1(u) - h_2(u)}_1,
\end{equation*}
\end{definition}

Under $\ell_{\infty,1}$-distance, we define the $r$-cover and the $r$-covering number for hypothesis class $\Hh$:
\begin{definition}[$r$-cover]
We say function class $\Hh_r$ $r$-covers $\Hh$ under $\ell_{\infty, 1}$-distance if for all function $h \in \Hh$, there exists $h_r$ in $\Hh_r$ such that $\norm{h - h_r}_{\infty, 1} \le r$.
\end{definition}

\begin{definition}[$r$-covering number]
The $r$-covering number of $\Hh$, denoted by $\Nn_{\infty,1}(\Hh, r)$, is the cardinality of the smallest function class $\Hh_r$ that $r$-covers $\Hh$ under $\ell_{\infty, 1}$-distance.
\end{definition}

We then derive the generalization bound of NE approximators. 
It describes the gap between the NE approximator's true risk $L_\Dd(h)$ and empirical risk $L_S(h)$ on the training set $S$.

\begin{restatable}{theorem}{ThmGB}[Generalization bound]
\label{theorem:GB}
For hypothesis class $\Hh$ of NE approximator and distribution $\Dd$, with probability at least $1 - \delta$ over draw of the training set $S$ from $\Dd$, $\forall h \in \Hh$ we have
\begin{equation*}
	L_\Dd(h) - L_S(h) \le  2\Delta_m
	+ 4\sqrt{\frac{2\ln(4/\delta)}{m}}
\end{equation*}
where $\Delta_m \coloneqq \inf_{r>0}\{ \sqrt{\frac{2\ln\Nn_{\infty,1}(\Hh, r)}{m}}+ 2r \}$.
\end{restatable}

\cref{theorem:GB} is quite general and orthogonal to the learning algorithm we use. 
It characterizes the generalization ability of all the NE approximators in normal-form games with fixed action space.
As we can see, with a large enough training set, the bound goes to zero (if $\Nn_{\infty, 1}(\Hh, r)$ is bounded) so that we can estimate the true risk through the empirical risk.

\subsection{Agnostic PAC Learnable}
\label{sec:agnostic_PAC}
If for arbitrary $r>0$ the covering number $\Nn_{\infty,1}(\Hh, r)$ can be bounded, then the bound in \cref{theorem:GB} goes to zero as the training set size $m \to \infty$.
Inspired by this, we make the following assumption to limit the representativeness of $\Hh$:
\begin{assumption}
\label{asp:H:cover}
For hypothesis class $\Hh$, we assume the logarithm of its $r$-covering number grows as a polynomial with respect to $1/r$. i.e., 
\begin{equation*}
	\ln \Nn_{\infty, 1}(\Hh, r) \le  \mathtt{Poly}(\frac{1}{r})
\end{equation*}
for $r > 0$.
\end{assumption}

\cref{asp:H:cover} is a standard assumption in PAC analysis~\citep{anthony1999neural}. 
It holds for many widely used machine learning models, including the classical linear model~\citep{zhang2002covering} and kernel method~\citep{zhou2002covering}.
Moreover, as we will prove, \cref{asp:H:cover} also holds for the \Lip hypothesis class, which includes neural networks with parameters of bounded ranges~\citep{szegedy2013intriguing,scaman2018lipschitz}.

\begin{definition}[\Lip hypothesis class]
\label{H:Lip}
We say $\Hh$ is a \Lip hypothesis class if there is a constant $L_\Hh>0$ such that for each function $h \in \Hh$ and game utility $u, v \in \Uu$, we have $\|h(u) - h(v) \|_{1} \le L_\Hh\|u-v\|_{\max}$,
\end{definition}

\begin{restatable}{lemma}{LipCover}
\label{Lip:cover}
\cref{asp:H:cover} holds For \Lip hypothesis class $\Hh$ since we have 
\begin{equation*}
    \Nn_{\infty, 1}(\Hh, r) \le O\left((\frac{L_\Hh}{r})^{\Dim}\ln\frac{1}{r}\right).
\end{equation*}
\end{restatable}

Based on \cref{asp:H:cover}, \cref{lemma:L_sigma} and \cref{lemma:L_u}, we prove the \emph{uniform convergence} of hypothesis class $\Hh$ with respect to Nash approximation loss.
It characterizes the sample complexity to probably obtain an \emph{$\epsilon$-representative} training set $S$.
That is, for an arbitrary function $h\in \Hh$, the empirical risk $L_S(h)$ on $S$ is close to the true risk $L_\Dd(h)$ up to $\epsilon$.

\begin{restatable}{theorem}{ThmUC}[Uniform convergence]
\label{theorem:UC}
Fix $\epsilon, \delta \in (0, 1)$, for hypothesis class $\Hh$ and distribution $\Dd$, with probability at least $1 - \delta$ over draw of the training set $S$ with
\begin{equation*}
	m \ge m_\Hh^{UC}(\epsilon, \delta) := \frac{9}{2\epsilon^2}\left(\ln\frac{2}{\delta} + \ln\Nn_{\infty,1}(\Hh, \frac{\epsilon}{6})\right)
\end{equation*}
games from $\Dd$, we have 
\begin{equation*}
    |L_S(h) - L_\Dd(h)| \le \epsilon
\end{equation*}
for all $h \in \Hh$.
$m_\Hh^{UC}(\epsilon, \delta)$ grows as a polynomial of $1/\epsilon$ and $\ln(1/\delta)$ under \cref{asp:H:cover}.
\end{restatable}

\cref{theorem:UC} is the sufficient condition for \emph{agnostic PAC learnable}, which provides the learnability guarantee of predicting NE from samples.

\begin{restatable}{theorem}{ThmPAC}
\label{theorem:PAC}
Fix $\epsilon, \delta \in (0, 1)$, for hypothesis class $\Hh$ and distribution $\Dd$, with probability at least $1 - \delta$ over draw of the training set $S$ with
\begin{equation*}
	m \ge m_\Hh(\epsilon, \delta) := \frac{18}{\epsilon^2}\left(\ln\frac{2}{\delta} + \ln\Nn_{\infty,1}(\Hh, \frac{\epsilon}{6})\right)
\end{equation*}
games from $\Dd$, when running empirical risk minimization on Nash approximation loss, we have 
\begin{equation*}
    L_\Dd(\mathrm{ERM}_\Hh(S)) \le \min_{h \in \Hh}L_\Dd(h) + \epsilon.
\end{equation*}
The sample complexity $m_\Hh(\epsilon, \delta)$ grows as a polynomial of $1/\epsilon$ and $\ln(1/\delta)$ under \cref{asp:H:cover}.
\end{restatable}

\cref{theorem:PAC} provides the (agnostic) PAC learnability of NE.
Under \cref{asp:H:cover}, when using a training set with size larger than a polynomial of $1/\epsilon$ and $\ln(1/\delta)$, with probability at least $1 - \delta$ the learned NE approximator can reach the near-optimal performance in $\Hh$ up to $\epsilon$.
As we will demonstrate by experiments in \cref{sec:experiments}, even equipped with the most simple neural architectures, the learned NE approximator can efficiently compute approximate NE solutions for games under the same distribution.

\begin{remark}
While realizability assumption~\citep{shalev2014understanding}, i.e., the assumption that $\min_{h \in \Hh} L_\Dd(h) = 0$, is adopted in many  PAC analyses~\citep{krishnamurthy2016pac, jin2021bellman}, however, it is not feasible in our case.
Due to the non-smooth issue, we discussed in \cref{tab:non-smooth}, it remains an open question whether there is a hypothesis class that satisfies the realizability assumption with limited complexity.
As a result, we consider \emph{agnostic} PAC learnability.
\end{remark}

\section{Experiments and Application}
\label{sec:experiments}

In this section, we first provide numerical experiments to verify the practicality of our PAC result.
Specifically, we construct a parameterized model as our hypothesis class and train an NE approximator via \cref{alg:training}.
We show that the learned NE approximator is computation-efficient with low generation loss.
Afterward, we demonstrate an application for the NE approximator: It can warm-start other NE solvers in bimatrix games by providing effective initializing points.
All of our experiments are run on a Linux machine with $48$ core Intel(R) Xeon(R) CPU (E5-2650 v4@2.20GHz) and $4$ TITAN V GPU.
Each experiment is run by $5$ times, and the average results are presented.

\subsection{Experimental Setup} 
We use GAMUT\footnote{\url{http://gamut.stanford.edu/}}~\citep{nudelman2004run}, a suite of game generators designated for testing game-theoretic algorithms, to generate the game instances.
We select $5$ game classes as our data distribution since they are nontrivial for TS algorithm~\citep{TS0.3393} to solve~\citep{fearnley2015empirical}: 
\begin{itemize}
\item
\textit{TravelersDilemma:} 
Each player simultaneously requests an amount of money and receives the lowest of the requests submitted by all players.

\item \textit{GrabTheDollar:} A price is up for grabs, and both players have to decide when to grab the price. The action of each player is the chosen times. 
If both players grab it simultaneously, they will rip the price and receive a low payoff.
If one chooses a time earlier than the other, she will receive the high payoff, and the opposing player will receive a payoff between the high and the low.

\item \textit{WarOfAttrition:} In this game, both players compete for a single object, and each chooses a time to concede the object to the other player. 
If both concede at the same time, they share the object. 
Each player has a valuation of the object, and each player's utility is decremented at every time step.

\item \textit{BertrandOligopoly:} All players in this game are producing the same item and are expected to set a price at which to sell the item. The player with the lowest price gets all the demand for the item and produces enough items to meet the demand to obtain the corresponding payoff.

\item 
\textit{MajorityVoting:} 
This is an $n$-player symmetric game.
All players vote for one of the $|A_1|$ candidates. 
Players' utilities for each candidate being declared the winner are arbitrary. 
If there is a tie, the winner is the candidate with the lowest number. 
There may be multiple Nash equilibria in this game.
\end{itemize}

For bimatrix games, we set the game size as $300\times 300$.
For multiplayer games, we generate the $3$ and $4$ player versions of \textit{BertrandOligopoly} and \textit{MajorityVoting} (The suffix -$3$ and -$4$ represent the $3$ and $4$ player versions, respectively).
We set the game size as $30\times 30\times 30$ for $3$-player games and $15\times 15\times 15\times 15$ for $4$-player games.
For each game class, we generate $2\times 10^4$ game instances with different random seeds, and we randomly divide $2000$ and $200$ instances for validation and testing.

As for the NE approximator, we construct a fully connected neural network as the hypothesis class due to the universal approximation theorem of it~\citep{Hornik1989MultilayerFN}.
We apply ReLU as the activation function and add batch normalization (without learnable parameters) before the activation function.
We use $4$ hidden layers with $1024$ nodes of each layer in our neural network.
We learn our model using the Adam optimizer, and we restrict the parameters of our model in the range of $[0, 1]$.
By doing so, we make our model a Lipschitz hypothesis class so that it satisfies \cref{asp:H:cover}.

\subsection{Generalization and Efficiency}

\input{table/table1a}
\input{table/table2a}
\paragraph{Generalization}
\cref{tab:loss} and \cref{tab:loss2} report the average Nash approximation loss of the trained NE approximator. 
We observe that the Nash approximation loss in the test set is sufficiently small and much lower than the random solutions.
Such a comparison result inspires us to use the predicted solution as the initial point for classical solvers.
Moreover, we can also see a small gap between the training and testing performance, which gives the feasibility of estimating the true risk of the NE approximator through its empirical risk on the training set. 
It also verifies the generalization bound in \cref{theorem:GB}.

\input{table/table1b}
\input{table/table2b}
\paragraph{Efficiency}
Notice that the NE approximator has never seen the test game instances in training.
The approximate solution is obtained by just a simple feed-forward neural network computation.
As a result, it can be used to infer approximate solutions quickly.
To better demonstrate the efficiency of the trained NE approximator, we record the time and iterations traditional algorithms spent (on the test set) to reach the same performance.
We use the following algorithms: 
\begin{itemize}
\item
\textit{Fictitious play (FP)}~\citep{monderer1996fictitious}: The most well-known learning algorithm to approximate Nash equilibrium; 
\item
\textit{Regret matching (RM)}~\citep{hart2000simple}: Representative method of no-regret learning, and it leads to coarse correlated equilibrium.
\item
\textit{Replicator dynamics (RD)}~\citep{schuster1983replicator}: A system of differential equations that describe how a population of strategies, or replicators, evolve through time. 
\item
\textit{TS Algorithm} (TS): The algorithm proposed by~\citet{TS0.3393}. It reaches an approximation ratio $\epsilon=0.3393$ for bimatrix games.
\item 
\textit{DFM algorithm} (DFM): The algorithm proposed by~\citet{DFM1/3}. It is an improved version of the TS algorithm, and reaches the current best approximation ratio $\epsilon=1/3$ for bimatrix games.
\end{itemize}

During implementation, we use GPU to speed up the computation of the baselines.
TS and DFM algorithm cannot be accelerated by GPU, so we run them on CPU. 
For FP, RM and RD, we set the maximum number of iterations to $100000$ and terminate the algorithm once it reaches the same performance. 
TS and DFM algorithm terminates with probability $1$, so we stop them early if the same performance has already been reached.

We present the efficiency results of bimatrix game in \cref{tab:time} and multiplayer game in \cref{tab:time2}. 
While the NE approximator efficiently comes up with an approximate solution, the baseline methods spend much more time to reach the same performance.
Sometimes the learning approaches FP, RM and RD even fail to converge to the same performance as NE approximator.

\subsection{Application: Warm-Start Classical Solvers}
\begin{figure}
\centering
\begin{subfigure}[b]{0.45\textwidth}
	\centering
	\includegraphics[width=\textwidth]{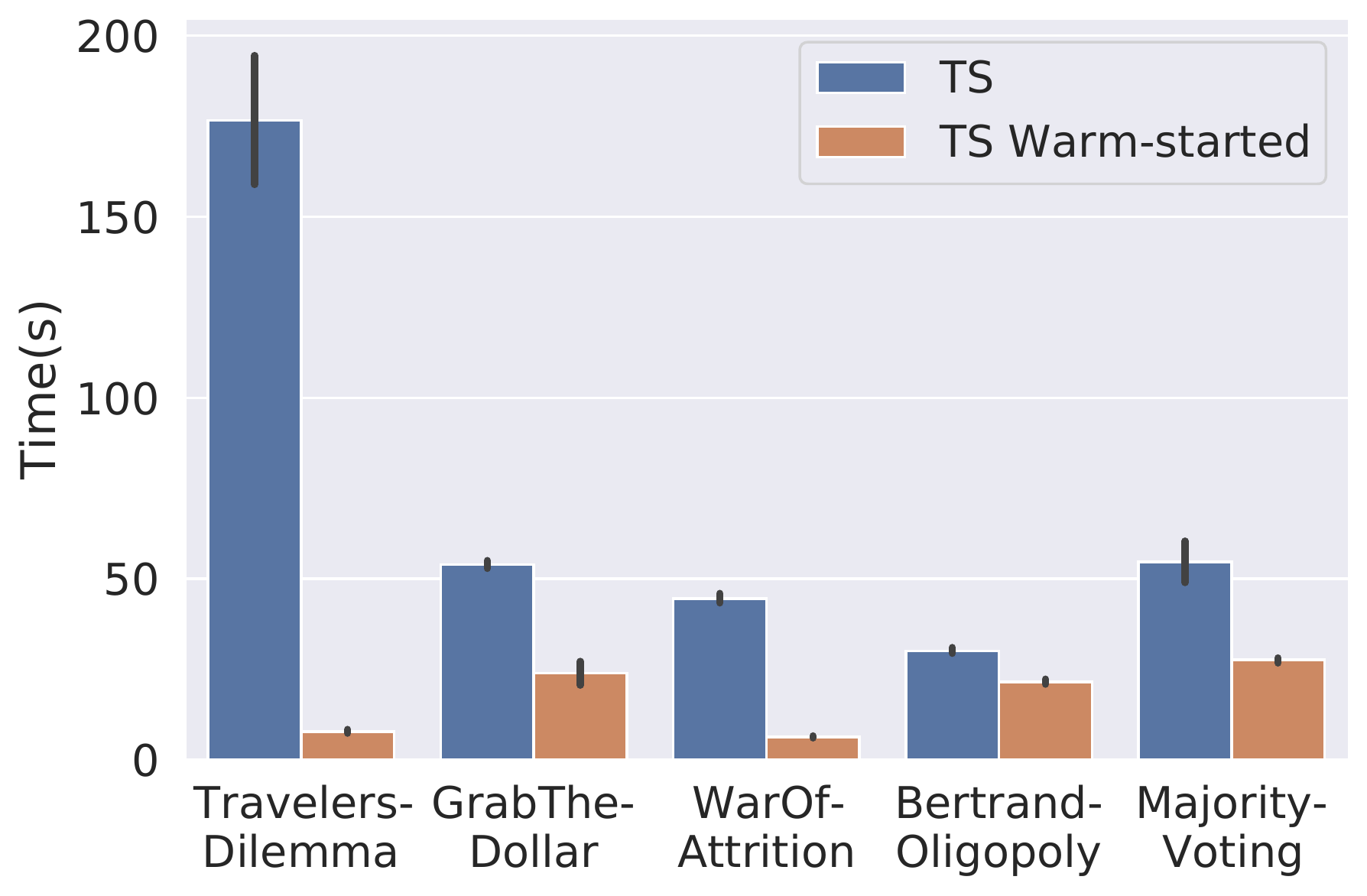}
	\caption{Warm-starting TS algorithm.
	}
	\label{fig:init_TS}
\end{subfigure}
\hfill
\begin{subfigure}[b]{0.45\textwidth}
	\centering
	\includegraphics[width=\textwidth]{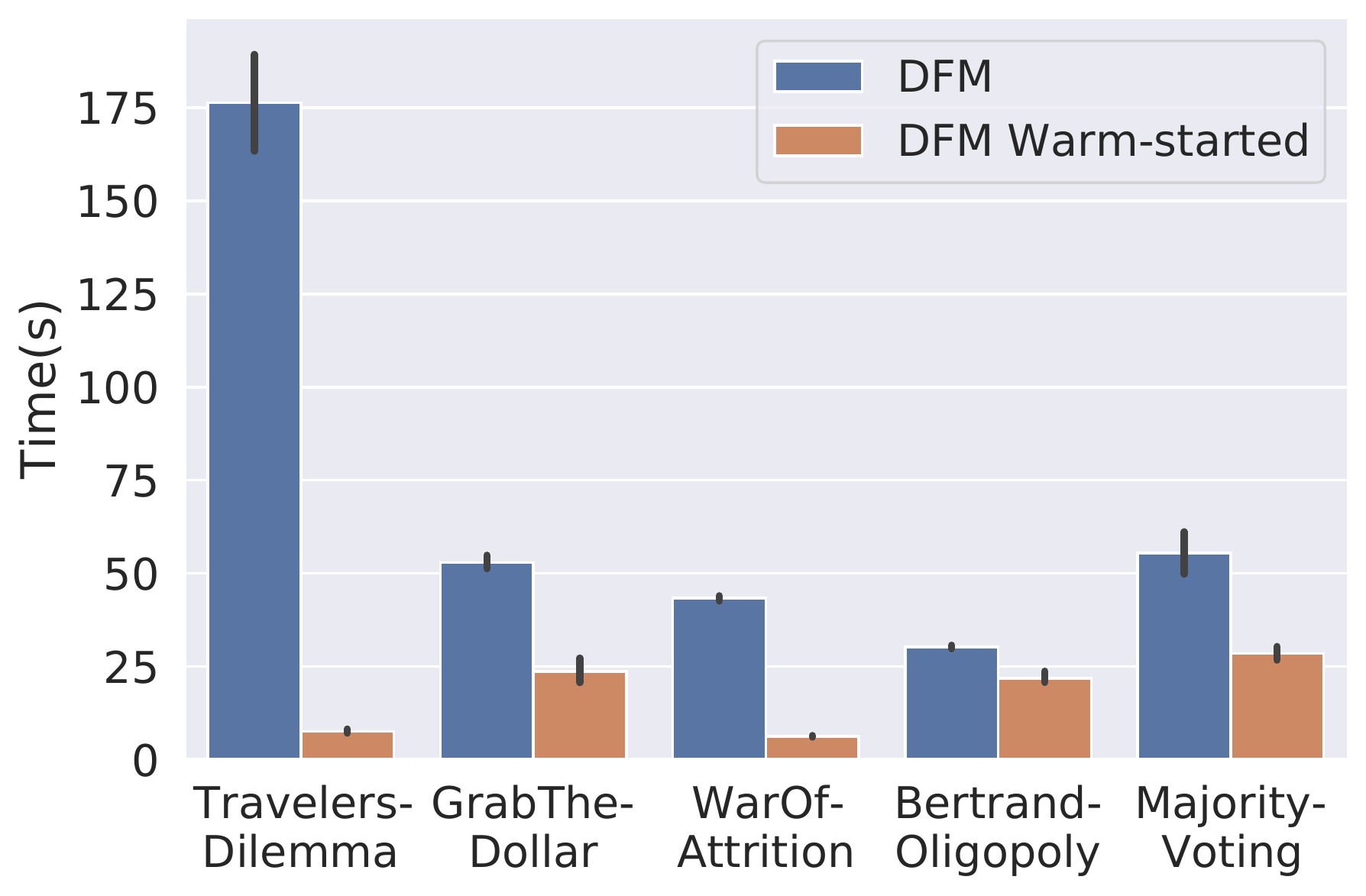}
	\caption{Warm-starting DFM algorithm.
	}
	\label{fig:init_DFM}
\end{subfigure}
\caption{
	Experimental results of warm-starting TS algorithm and DFM algorithm with NE approximator.  
	Each experiment is run by $5$ times. 
	Average results and $95\%$ confidence intervals are shown.
}
\end{figure}
As we can see from the previous experiments, the NE approximator could be efficient for the games under the same distribution.
Moreover, it can achieve a better Nash approximation loss than random solutions. 
Meanwhile, the classical NE solvers, such as the TS and DFM algorithm, usually set random strategies as the starting point.
Therefore, it is promising to warm-start those algorithms with the NE approximator.
By doing so, we benefit from both advantages of the function-approximation method and the traditional approach.
The NE approximator can infer initial solutions in batches with low computational costs, and the classical solvers can provide theoretical guarantees.

\cref{fig:init_TS} and \cref{fig:init_DFM} report the experimental results of warm-starting TS and DFM algorithm, respectively.
We can observe that by taking the output strategies of the NE approximator as the pre-solving initialization, both TS and DFM algorithms spend less time to terminate, especially in game \textit{TravelersDilemma} 
and \textit{WarOfAttrition}.
Notice that both algorithms ensure that the final solutions will be better than the initial solutions.
Thus, it would always be helpful to provide a good starting point for them.

\section{Conclusion}
\label{sec:conclusion}

In this paper, we study the learnability of predicting NE in $n$-player normal-form games with fixed action space.
Theoretically, we provide a generalization bound for the NE approximator under Nash approximation loss.
The bound gives a theoretical guarantee of the generalization ability.
We then prove that Nash equilibrium is agnostic PAC learnable.
Such a result provides the feasibility of obtaining a good NE approximator via empirical risk minimization.
Empirically, we conduct numerical experiments to verify the learned NE approximator's generalization ability and efficiency.
Afterward, we demonstrate the application of the NE approximator to warm-start other classical solvers, and we report fast convergence.
Our theoretical and empirical results show the practicability of learning an NE approximator via data-driven approach.
As for future work, we are interested in extending the learnability results of the NE approximator to settings beyond normal-form games, and exploring a more efficient hypothesis class for the NE approximator.

\section*{Acknowledgement}
This work is supported by the National Natural Science Foundation of China (Grant No. 62172012).
We thank Xiang Yan, Dongge Wang, David Mguni and Kun Shao for various helpful discussions.
We thank all anonymous reviewers for their helpful feedback.

%% file: table/table1a.tex
\begin{table*}[t]
    \centering
    \caption{
    	The average Nash approximation loss (and the corresponding standard deviation across random seeds) of the learned NE approximator on training and testing, compared with random solutions.
    	All the games are $300\times 300$ bimatrix games.
    }
    \begin{sc}
    \resizebox{\textwidth}{!}{
    \begin{tabular}{ccccccc}
        \toprule
        & \textit{TravelersDilemma} & \textit{GrabTheDollar} & \textit{WarOfAttrition} & \textit{BertrandOligopoly} & \textit{MajorityVoting}  \\
        & $\Nap$ & $\Nap$ & $\Nap$ & $\Nap$ & $\Nap$ \\ 
        \midrule \midrule
        Random & 0.2644 \scriptsize{$\pm$ 2.57e-4} & 0.2603 \scriptsize{$\pm$ 2.78e-4} & 0.3396 \scriptsize{$\pm$ 2.65e-4} & 0.3208 \scriptsize{$\pm$ 3.86e-4} & 0.4727 \scriptsize{$\pm$ 5.98e-4} \\
        \midrule
        Train & 1.013e-6 \scriptsize{$\pm$ 1.07e-7} & 8.328e-5 \scriptsize{$\pm$ 7.87e-5} & 2.984e-7 \scriptsize{$\pm$ 1.65e-8} & 3.402e-4 \scriptsize{$\pm$ 5.48e-5} & 4.416e-6 \scriptsize{$\pm$ 1.02e-6} \\
        Test & 
        0.991e-6 \scriptsize{$\pm$ 1.04e-7} & 4.823e-5 \scriptsize{$\pm$ 5.36e-5} & 2.871e-7 \scriptsize{$\pm$ 1.89e-8} & 3.338e-4 \scriptsize{$\pm$ 4.90e-5} & 5.526e-6 \scriptsize{$\pm$ 2.04e-6} \\
        \bottomrule
    \end{tabular}
    }
    \label{tab:loss}
    \end{sc}

\end{table*}

%% file: table/table2a.tex
\begin{table*}[t]
	\centering
	\caption{
    	The average Nash approximation loss (and the corresponding standard deviation across random seeds) of the learned NE approximator on training and testing, compared with random solutions.
	    The game dimension is $30\times 30\times 30$ for $3$-player games and $15\times 15\times 15\times 15$ for $4$-player games. 
	}
	\begin{sc}
	\resizebox{\textwidth}{!}{
	\begin{tabular}{ccccc}
		\toprule
		& BertrandOligopoly-3 & MajorityVoting-3 & BertrandOligopoly-4 &  MajorityVoting-4 \\
		& $\Nap$ & $\Nap$ & $\Nap$ & $\Nap$ \\ 
		\midrule \midrule
		Random & 0.1145 \scriptsize{$\pm$ 9.16e-4} & 0.3534 \scriptsize{$\pm$ 1.11e-3} & 0.0573 \scriptsize{$\pm$ 5.41e-4} & 0.2428 \scriptsize{$\pm$ 1.34e-3}\\
		\midrule 
		Train & 4.046e-6 \scriptsize{$\pm$ 7.22e-6} & 1.018e-3 \scriptsize{$\pm$ 3.45e-4} & 1.619e-7 \scriptsize{$\pm$ 3.51e-8} & 3.881e-4 \scriptsize{$\pm$ 1.45e-4} \\
		Test & 2.525e-6 \scriptsize{$\pm$ 4.20e-6} & 0.612e-3 \scriptsize{$\pm$ 2.54e-4} & 1.643e-7 \scriptsize{$\pm$ 3.53e-8} & 2.359e-4 \scriptsize{$\pm$ 4.36e-4}\\
		\bottomrule
	\end{tabular}
	}
	\label{tab:loss2}
	\end{sc}
	
\end{table*}

%% file: table/table1b.tex
\begin{table*}[t]
    \centering
    \caption{
    The average time and iterations traditional algorithms spent on each test set to reach the same performance as the NE approximator (NEA) in $300\times 300$ bimatrix games.
    $^*$ represents the method fails to reach the same performance under the limitation of the maximum iterations in some of the $5$ runs.
    }
    \resizebox{\textwidth}{!}{
    \begin{sc}
    \begin{tabular}{crrrrrrrrrr}
        \toprule
        & \multicolumn{2}{c}{\textit{TravelersDilemma}} & \multicolumn{2}{c}{\textit{GrabTheDollar}} & \multicolumn{2}{c}{\textit{WarOfAttrition}} & \multicolumn{2}{c}{\textit{BertrandOligopoly}} & \multicolumn{2}{c}{\textit{MajorityVoting}}  
        \\
        & Time & Iteration & Time & Iteration & Time & Iteration & Time & Iteration & Time & Iteration  \\ 
        \midrule \midrule
        FP & $^*$99.9s & $^*$100000 & $^*$95.2s & $^*$100000 & 57.6s & 59919.6 & 61.5s & 63794.8 & 77.1s & 80452.6 \\
        RM & 162.2s & 85442.2 & $^*$190.3s & $^*$98788.2 & 149.3s & 79151.2 & 28.1s & 14544.0 & $^*$189.4s & $^*$100000 \\
        RD & 4.3s & 3813.4 & 2.5s & 2212.4 & 2.8s & 2482.8 & 1.0s & 826.2 & $^*$118.7s & $^*$100000 \\
        TS & 149.9s & -- & 47.1s & -- & 41.9s & -- & 26.6s & -- & 44.1s & -- \\
        DFM & 147.6s & -- & 45.8s & -- & 39.8s & -- & 27.1s & -- & 44.7s & -- \\
        \midrule
        NEA & \textbf{$\bm <$0.5s} & \textbf{1.0} & \textbf{$\bm <$0.5s} & \textbf{1.0} & \textbf{$\bm <$0.5s} & \textbf{1.0} & \textbf{$\bm <$0.5s} & \textbf{1.0} & \textbf{$\bm <$0.5s} & \textbf{1.0} \\
        \bottomrule
    \end{tabular}
    \label{tab:time}
    \end{sc}
    }

\end{table*}

%% file: table/table2b.tex
\begin{table*}[t]
\centering
\caption{
The average time and iterations traditional algorithms spent on each test set to reach the same performance as the NE approximator (NEA) in $30\times 30\times 30$ and $15\times 15\times 15\times 15$ games.
$^*$ represents the method fails to reach the same performance under the limitation of the maximum iterations in some of the $5$ runs.
}
\resizebox{\textwidth}{!}{
\begin{sc}
\begin{tabular}{crrrrrrrr}
    \toprule
    \multirow{2}{*}{Methods} & \multicolumn{2}{c}{BertrandOligopoly-3} & \multicolumn{2}{c}{MajorityVoting-3} & \multicolumn{2}{c}{BertrandOligopoly-4} & \multicolumn{2}{c}{MajorityVoting-4}  \\
    & Time & Iteration & Time & Iteration & Time & Iteration & Time & Iteration \\ 
    \midrule \midrule
    FP & $^*$186.3s & $^*$100000 & 2.5s & 1273.8 & $^*$311.6s & $^*$100000 & 183.8s & 59324.0 \\
    RM & $^*$258.6s & $^*$100000 & 19.8s & 7395.6 & $^*$393.3s & $^*$100000 & 106.5s & 27331.6 \\
    RD & 33.0s & 28629.0 & 0.9s & 607.2 & 50.6s & 23100.2 & 143.3s & 61431.2 \\
    \midrule
    NEA & \textbf{$\bm <$0.5s} & \textbf{1.0} & \textbf{$\bm <$0.5s} & \textbf{1.0} & \textbf{$\bm <$0.5s} & \textbf{1.0} & \textbf{$\bm <$0.5s} & \textbf{1.0} \\
    \bottomrule
\end{tabular}
\label{tab:time2}
\end{sc}
}
\end{table*}

%% file: _appendix.tex
\clearpage
\appendix
\onecolumn

\section{Omitted Proofs}
\label{sec:proof}

\subsection{Proof of \cref{lemma:L_sigma}}
\label{sec:proof:L_sigma}
\LemLSigma*
\begin{proof}
	$\forall \sigma, \sigma'$, we define $y_{-j} \coloneqq (\sigma_1, \dots, \sigma_{j-1}, \sigma'_{j+1}, \dots, \sigma'_n)$. 
	Then, $\forall i \in N$ we have
	\begin{equation*}
		\begin{aligned}
			|u_i(\sigma) - u_i(\sigma')|
			=& |u_i(\sigma_1, \sigma_2, \dots, \sigma_n) - u_i(\sigma'_1, \sigma'_2, \dots, \sigma'_n)| 
			\\
			=& \Big| \sum_{j=1}^n \Big(u_i(\sigma_1, \dots, \sigma_j, \sigma'_{j+1}, \dots, \sigma'_n) 
			- u_i(\sigma_1, \dots, \sigma'_j, \sigma'_{j+1}, \dots, \sigma'_n) \Big)\Big|, 
			\\
			=& \Big| \sum_{j=1}^n \Big(u_i(\sigma_j, y_{-j}) - u_i(\sigma'_j, y_{-j}) \Big)\Big|
			\\
			=& \Big| \sum_{j=1}^n \sum_{a_j}(\sigma_j(a_j) - \sigma'_j(a_j)) \sum_{a_{-j}} u_i(a_j, a_{-j})y_{-j}(a_{-j})  \Big|
			\\
			\le& \sum_{j=1}^n \sum_{a_j}\Big|\sigma_j(a_j) - \sigma'_j(a_j)\Big|\sum_{a_{-j}} u_i(a_j, a_{-j})y_{-j}(a_{-j}) 
			\\
			\overset{(a)}{\le}& \sum_{j=1}^n \sum_{a_j}\Big|\sigma_j(a_j) - \sigma'_j(a_j)\Big|\sum_{a_{-j}}y_{-j}(a_{-j})
			\\
			\le& \sum_{j=1}^n \sum_{a_j\in A_j}\Big|\sigma_j(a_j) - \sigma'_j(a_j)\Big|
			= \norm{\sigma - \sigma'}_1,
		\end{aligned}
	\end{equation*}
	where $(a)$ holds since $u_i(\cdot) \in [0, 1]$.
	Therefore, $\forall a_i \in A_i$, 
	\begin{equation*}
		\begin{aligned}
			u_i(a_i, \sigma_{-i}) - u_i(\sigma) =& u_i(a_i, \sigma_{-i}) - u_i(a_i, \sigma'_{-i}) + u_i(a_i, \sigma'_{-i}) - u_i(\sigma') + u_i(\sigma') - u_i(\sigma)
			\\
			\le& \norm{\sigma - \sigma'}_1 + \Nap(\sigma', u) + \norm{\sigma - \sigma'}_1
			\\
			=& \Nap(\sigma', u) + 2\norm{\sigma - \sigma'}_1.
		\end{aligned}
	\end{equation*}
	Based on that, we get
	\begin{equation*}
		\begin{aligned}
			\Nap(\sigma, u) 
			=& \max_{i\in N, a_i \in A_i}[u_i(a_i, \sigma_{-i}) - u_i(\sigma)]
			\\
			\le& \Nap(\sigma', u) + 2\norm{\sigma - \sigma'}_1
		\end{aligned}
	\end{equation*}
	Similarly, we also have
	\begin{equation*}
		\Nap(\sigma', u) 
		\le \Nap(\sigma, u) + 2\norm{\sigma - \sigma'}_1
	\end{equation*}
\end{proof}

\subsection{Proof of \cref{lemma:L_u}}
\label{sec:proof:L_u}
\LemLU*
\begin{proof}
$\forall u, v \in \Uu, \sigma \in \Delta A_1 \times A_2 \times \dots \times A_n, i \in N, a_i \in A_i$, we have
\begin{equation*}
	\begin{aligned}
		u_i(a_i, \sigma_{-i}) =& v_i(a_i, \sigma_{-i}) + (u_i(a_i, \sigma_{-i}) - v_i(a_i, \sigma_{-i})) 
		\\
		\le& v_i(a_i, \sigma_{-i}) + \|u-v\|_{\max},
		\\
		\le& v_i(\sigma) + \Nap(\sigma, v) + \|u-v\|_{\max},
		\\
		\le& u_i(\sigma) + \Nap(\sigma, v) + 2\|u-v\|_{\max}
	\end{aligned}
\end{equation*}
Therefore,
\begin{equation*}
	\begin{aligned}
		\Nap(\sigma, u) &= \max_{i\in N, a_i \in A_i}[u_i(a_i, \sigma_{-i}) - u_i(\sigma)] \le \Nap(\sigma, v) + 2\|u-v\|_{\max}
	\end{aligned}
\end{equation*}
Similarly, we also have
\begin{equation*}
	\Nap(\sigma, v) \le \Nap(\sigma, u) + 2\|u-v\|_{\max}
\end{equation*}
\end{proof}

\subsection{Proof of \cref{theorem:GB}}
\ThmGB*
To prove \cref{theorem:GB}, we use an auxiliary lemma from \citet{shalev2014understanding}.
We measure the capacity of the composite function class $\Nap \circ \Hh$ using the empirical Rademacher complexity on the training set $S$, which is defined as:
\begin{equation*}
\begin{aligned}
	\Rr_S(\Nap \circ \Hh) \coloneqq \frac{1}{m}\EE_{\bm x \sim \{+1,-1\}^m}\Big[\sup_{h\in\Hh} \sum_{i=1}^m x_i \cdot \Nap(h(u^{(i)}), u^{(i)}) \Big],
\end{aligned}
\end{equation*}
where $\bm x$ is distributed i.i.d. according to uniform distribution in $\{+1,-1\}$.
We have

\begin{lemma}[\citet{shalev2014understanding}]
\label{lemma:Rad}
Let $S$ be a training set of size $m$ drawn i.i.d. from distribution $\Dd$ over $\Uu$.
Then with probability at least $1 - \delta$ over draw of $S$ from $\Dd$, for all $h \in \Hh$,
\begin{equation*}
	L_\Dd(h) - L_S(h) \le 2\Rr_S(\Nap \circ \Hh) + 4\sqrt{\frac{2\ln(4/\delta)}{m}}
\end{equation*}
\end{lemma}

\begin{proof}[Proof of \cref{theorem:GB}]
For hypothesis class $\Hh$, let $\Hh_r$ with $|\Hh_r| = \Nn_{\infty,1}(\Hh, r)$ be the function class that $r$-covers $\Hh$ for some $r>0$. 
$\forall h \in \Hh$, denote $h_r \in \Hh_r$ be the function approximator that $r$-covers $h$.
Based on \cref{lemma:L_sigma}, we have
\begin{equation}
	\begin{aligned}
		\label{eq:L_sigma_r} 
		|\Nap(h(u), u) - \Nap(h_r(u), u)| \le 2\|h(u) - {h}_r(u)\|_1 \le 2r
	\end{aligned}
\end{equation}

We thus have
\begin{equation}
	\label{eq:Rad}
	\begin{aligned}
		\Rr_S(&\Nap \circ \Hh) 
		= \frac{1}{m}\EE_{\bm x}\Big[\sup_h \sum_{i=1}^m x_i \cdot \Nap(h(u^{(i)}), u^{(i)}) \Big]
		\\
		=& \frac{1}{m}\EE_{\bm x}\Big[\sup_h \sum_{i=1}^m x_i \cdot \big(\Nap(h_r(u^{(i)}), u^{(i)}) 
		\\
		&+ \Nap(h(u^{(i)}), u^{(i)})  - \Nap(h_r(u^{(i)}), u^{(i)})\big)  \Big]
		\\
		\overset{(a)}{\le}& \frac{1}{m}\EE_{\bm x}\Big[\sup_{h_r \in \Hh_r} \sum_{i=1}^m x_i \cdot \Nap(h_r(u^{(i)}), u^{(i)})  \Big]
		\\
		&+ \frac{1}{m}\EE_{\bm x}\Big[\sup_{h\in \Hh} \sum_{i=1}^m |x_i \cdot 2r| \Big]
		\\
		\overset{(b)}{\le}& \sup_{h_r\in \Hh_r}\sqrt{\sum_{i=1}^m \ell^2(h_r, u^{(i)})} \cdot \frac{\sqrt{2\ln\Nn_{\infty,1}(\Hh, r)}}{m} + \frac{2r}{m}\EE_{\bm x}\norm{\bm x}_1
		\\
		\le& \sqrt{\frac{2\ln\Nn_{\infty,1}(\Hh, r)}{m}}+ 2r
	\end{aligned}
\end{equation}
where the second term of $(a)$ holds from \cref{eq:L_sigma_r}, and the first term of $(b)$ holds by Massart's lemma~\citep{shalev2014understanding}.

Combining \cref{lemma:Rad} and \cref{eq:Rad}, we get

\begin{equation*}
	L_\Dd(h) - L_S(h) \le 2\cdot \inf_{r>0}\Big\{ \sqrt{\frac{2\ln\Nn_{\infty,1}(\Hh, r)}{m}}+ 2r \Big\}
	+ 4\sqrt{\frac{2\ln(4/\delta)}{m}}
\end{equation*}
\end{proof}

\subsection{Proof of \cref{Lip:cover}}
\label{sec:L_H}
\LipCover*
\begin{proof}

The proof is done by construction.

\paragraph{Construct 1:} A $\nu$-covering set $\Uu_\nu \subset \Uu$ for $\Uu$ with respect to $\ell_{\max}$ distance.

We do so by discretizing 
each element along $[0, 1]$ at scale $\nu$. 
We discretize it into $\{0, \nu, 2\nu, \dots,  \lfloor \frac{1}{\nu} \rfloor \nu\}$.
Therefore, we have 
$
	|\Uu_\nu| = \left\lceil \frac{1}{\nu} \right\rceil^{\Dim}
$.

\paragraph{Construct 2:} A $\mu$-covering set $\Pi_\mu = \Pi_{\mu, 1} \times \dots \Pi_{\mu, n} \subset \Delta A_1 \times \dots \times \Delta A_n$ with respect to $\ell_1$ distance for the range of $\Hh$. For all $i \in N$, $\Pi_{\mu, i} \in \Delta A_i$ is a $\frac{\mu}{n}$-covering set of $\Delta A_i$ w.r.t. $\ell_1$ distance.

First, we define $g(x)$ for $x \in (0, 1)$ as the maximum value $y \in (0, x]$ such that $\frac{1}{y}$ is an integer. From the definition, we can easily get the following corollary:
\begin{corollary}
	$g(x) \in (\frac{x}{10}, x]$ for $x \in (0, 1)$.
\end{corollary}
We construct $\Pi_{\mu, i}$ by discretizing each element  along $[0, 1]$ at scale $g(\frac{\mu}{n|A_i|})$. 
By doing so, $\Pi_{\mu, i}$ is a $\frac{\mu}{n}$-covering set of $\Delta A_i$. 
The cardinal number of $\Pi_{\mu, i}$ is a combination number: 
\begin{equation*}
	\begin{aligned}
		|\Pi_{\mu, i}| =& \binom{(g(\frac{\mu}{n|A_i|}))^{-1} + |A_i|}{|A_i| - 1}
		\le \left(\frac{e(g(\frac{\mu}{n|A_i|}))^{-1} + e|A_i|}{|A_i| - 1}  \right)^{|A_i| - 1} 
		< \left(\frac{e(\frac{10n|A_i|}{\mu}) + e|A_i|}{|A_i| - 1}  \right)^{|A_i| - 1}, 
	\end{aligned}
\end{equation*}
So that 
\begin{equation*}
	\begin{aligned}
		|\Pi_\mu| = \prod_{i\in N} |\Pi_{\mu, i}|
		= \prod_{i\in N} \binom{(g(\frac{\mu}{n|A_i|}))^{-1} + |A_i|}{|A_i| - 1}
		\le \prod_{i\in N} \left(\frac{e(\frac{10n|A_i|}{\mu}) + e|A_i|}{|A_i| - 1}  \right)^{|A_i| - 1}
		= O((\frac{1}{\mu}) ^{\sum_{i \in N}|A_i| - n})
	\end{aligned}
\end{equation*}

\paragraph{Construct 3:} A $r$-covering set $\Hh_r$ of $\Hh$ with respect to $\ell_{\infty, 1}$-distance.

For all $u \in \Uu$, we define 
$u_\nu \in \Uu_\nu$ as the closed utility matrix to $u$ in $\Uu_\nu$ (with arbitrary tie-breaking rule) so that we have $\norm{u_\nu - u}_{\max} \le \nu$.
Based on this, we construct an auxiliary function class $\Ff_{\nu, \mu}$.
It contains all the functions $f: \Uu \to \Pi_\nu$ that satisfy $f(u) = f(u_\nu)$ for all $u \in \Uu$. By the definition of $\Ff_{\nu, \mu}$, we have
\begin{equation*}
	|\Ff_{\nu, \mu}| = |\Pi_\mu|^{|\Uu_\nu|}
\end{equation*}

For all $h \in \Hh$, denote $h_{\nu, \mu}$ as the closed function to $h$ in $\Ff_{\nu, \mu}$ with respect to $\ell_{\infty, 1}$-distance. We have $\norm{h - h_{\nu, \mu}}_{\infty,1} \le \mu$. Then we have

\begin{equation}
	\label{eq:L_H:nu:mu}
	\begin{aligned}
		|\Nap(h(u), u) - \Nap(h_{\nu, \mu}(u), u)| \le& 2\|h(u) - h_{\nu, \mu}(u)\|_1
		\\
		=& 2\|h(u) - h_{\nu, \mu}(u_\nu)\|_1
		\\
		\le& 2\|h(u) - h(u_\nu)\|_1 + 2\|h(u_\nu) - h_{\nu, \mu}(u_\nu)\|_1
		\\
		\le& 2L_\Hh\nu + 2\mu
	\end{aligned}
\end{equation}

Let $\Hh_r = \Ff_{\frac{r}{4L_\Hh}, \frac{r}{4}}$. According to \cref{eq:L_H:nu:mu}, $\Hh_r$ is a $r$-covering set of $\Hh$ with respect to $\ell_{\infty, 1}$ distance. 
Therefore,
\begin{equation*}
	\begin{aligned}
		\ln\Nn_{\infty, 1}(\Hh, r) \le& \ln|\Hh_r| = |\Uu_{\frac{r}{4L_\Hh}}|\ln|\Pi_{\frac{r}{4}}|
		\le \left\lceil \frac{4L_\Hh}{r} \right\rceil^{\Dim} \sum_{i\in N}(|A_i| - 1)\ln\left(\frac{e(\frac{40n|A_i|}{r}) + e|A_i|}{|A_i| - 1} \right)
		\\
		=& O\left((\frac{L_\Hh}{r})^{\Dim}\ln\frac{1}{r}\right)
	\end{aligned}
\end{equation*}
\end{proof}

\subsection{Proof of \cref{theorem:UC}}
\ThmUC*
\begin{proof}
	$\forall \epsilon \in (0, 1)$, we set the covering radius $r=\frac{\epsilon}{6}$. 
	Then,
	\begin{equation*}
		\begin{aligned}
			&\PP_{S\sim\Dd^m}\Big[\exists h\in\Hh, \big|L_S(h)
			- L_\Dd(h)\big| > \epsilon \Big] 
			\\
			\le& \PP_{S\sim\Dd^m}\Big[ \exists h\in\Hh, \big|L_S(h) - L_S({h}_r)\big| + \big| L_S({h}_r) - L_\Dd({h}_r)\big|
			+ \big|L_\Dd({h}_r) - L_\Dd(h)\big| 
			> \epsilon \Big] 
			\\
			\overset{(a)}{\le}& \PP_{S\sim\Dd^m}\Big[\exists h\in\Hh, 
			2r + \big| L_S({h}_r) - L_\Dd({h}_r)\big| + 2r > \epsilon \Big] 
			\\
			\le& \PP_{S\sim\Dd^m}\Big[\exists {h}_r \in {\Hh}_r, \big|L_S({h}_r) - L_\Dd({h}_r) \big| > \frac{1}{3}\epsilon \Big]
			\\
			\overset{(b)}{\le}& \Nn_{\infty,1}(\Hh, \frac{\epsilon}{6})\PP_{S\sim\Dd^m}\Big[\big|L_S({h}) - L_\Dd({h}) \big| > \frac{1}{3}\epsilon \Big]
			\\
			\overset{(c)}{\le}& 2\Nn_{\infty,1}(\Hh, \frac{\epsilon}{6})\exp(-\frac{2}{9}m\epsilon^2),
		\end{aligned}
	\end{equation*}
	where $(a)$ holds by \cref{eq:L_sigma_r}, $(b)$ holds by union bound, and $(c)$ holds by Hoeffding inequality.
	
	As a result, when $m \ge  \frac{9}{2\epsilon^2}\left(\ln\frac{2}{\delta} + \ln\Nn_{\infty,1}(\Hh, \frac{\epsilon}{6})\right)$, we have
	\begin{equation*}
		\begin{aligned}
			\PP_{S\sim\Dd^m}\Big[\exists h\in \Hh, \Big|L_S(h) - L_\Dd(h)\Big| > \epsilon \Big] < \delta
		\end{aligned}
	\end{equation*}
\end{proof}

\subsection{Proof of \cref{theorem:PAC}}
\label{sec:proof:PAC}
\ThmPAC*
\begin{proof}
	Define $m_\Hh(\epsilon, \delta) := m_\Hh^{UC}(\epsilon/2, \delta)$.
	According to \cref{theorem:UC}, 
	when $m \ge m_\Hh(\epsilon, \delta)$, with probability at least $1-\delta$ the training set $S$ is $\epsilon/2$-representative.
	Therefore, denote $h_S = \mathrm{ERM}_\Hh(S)$ and $h^* \in \arg\min_{h \in \Hh}L_\Dd(h)$, with probability at least $1-\delta$ we have
	\begin{equation*}
		\label{eq:proof:PAC}
		\begin{aligned}
			L_\Dd(h_S) \le& L_S(h_S) + \frac{\epsilon}{2}
			\le L_S(h^*) + \frac{\epsilon}{2}
			\le L_\Dd(h^*) + \frac{\epsilon}{2} + \frac{\epsilon}{2}
			= L_\Dd(h^*) + \epsilon
		\end{aligned}    
	\end{equation*}
\end{proof}